\documentclass{llncs}

\title{Online Pattern Matching for String Edit Distance with Moves\thanks{This work was supported by JSPS KAKENHI(24700140,26280088) and the JST PRESTO program}}

\author{Yoshimasa Takabatake\inst{1} 
  \and Yasuo Tabei\inst{2}
  \and Hiroshi Sakamoto\inst{1}}
\institute{
  Kyushu Institute of Technology~\email{\{takabatake,hiroshi\}@donald.ai.kyutech.ac.jp}
  \and 
  PRESTO, Japan Science and Technology Agency~\email{tabei.y.aa@m.titech.ac.jp}
}

\usepackage{makeidx}
\usepackage{algorithm}
\usepackage{algorithmicx}
\usepackage{algpseudocode}
\usepackage{amsmath}
\usepackage{url}
\usepackage[dvips]{graphicx}
\usepackage{comment}

\newtheorem{prob}{Problem}

\begin{document}

\maketitle

\begin{abstract}
Edit distance with moves (EDM) is a string-to-string distance measure that includes substring moves 
in addition to ordinal editing operations to turn one string to the other.
Although optimizing EDM is intractable, it has many applications especially in error detections. 
Edit sensitive parsing (ESP) is an efficient parsing algorithm that 
guarantees an upper bound of parsing discrepancies between different appearances 
of the same substrings in a string. ESP can be used for computing an approximate EDM as the $L_1$ distance between 
characteristic vectors built by node labels in parsing trees. 
However, ESP is not applicable to a streaming text data where a whole text is unknown in advance. 
We present an online ESP (OESP) that enables an online pattern matching for EDM. 
OESP builds a parse tree for a streaming text and computes the $L_1$ distance between characteristic vectors
in an online manner.
For the space-efficient computation of EDM, OESP directly encodes the parse tree into a succinct representation 
by leveraging the idea behind recent results of a dynamic succinct tree. 
We experimentally test OESP on the ability to compute EDM in an online manner on benchmark datasets, 
and we show OESP's efficiency.
\end{abstract}

\section{Introduction}
Streaming text data appears in many application domains of information retrieval. 
Social data analysis faces a problem for analyzing continuously generated texts.
In computational biology, recent sequencing technologies enable us to sequence individual genomes in a short time, 
which resulted in generating a large collection of genome data. 
There is therefore a strong incentive to develop a powerful method for analyzing streaming texts on a large-scale. 

{\em Edit distance with moves~(EDM)} is a string-to-string distance measure that includes substring moves 
in addition to insertions and deletions to turn one string to the other in a series of editing operations.
The distance measure is motivated in error detections, e.g., insertions and deletions on 
lossy communication channels~\cite{Levenshtein96}, typing errors in documents~\cite{Crochemore94} and evolutionary changes 
in biological sequences~\cite{Durbin98}.                             
Computing an optimum solution of EDM is intractable, 
since the problem is known to be NP-complete~\cite{Muthukrishnan00}. 
Therefore, researchers have paid considerable efforts to develop efficient approximation algorithms 
that are only applicable to an offline case where a whole text is given in advance (Table~\ref{tbl:comp}).  
Early results include the reversal model~\cite{Kececioglu1993,Bafna1996} which takes a substring of 
unrestricted size and replaces it by its reverse in one operation. 
Muthukrishnan and Sahinalp~\cite{Muthukrishnan00} proposed 
an approximate nearest neighbor considered as a sequence comparison with block operations. 
Recently, Shapira and Storer proposed a polylog time algorithm with $O(\lg{N}\lg^*{N})$ approximation ratio 
for the length $N$ of an input text. 


\begin{table}[t]
\begin{center}
{\scriptsize
\caption{Summary of recent pattern matching methods for EDM. The table summaries upper bound for the approximation ratio of EDM, computation time and space for each method.
The space for ESP and OESP is presented in bits. 
$N$ is the length of an input string; $\sigma$ is the alphabet size; $n$ is the number of variables in CFG; $\alpha \in (0,1]$ is a parameter for a hash table; 
$\lg^*$ is the iterated logarithm;
$\lg$ stands for $\log_2$.
}
}
\vspace{-0.5cm}
\label{tbl:comp}
{\scriptsize
\begin{tabular}{l|c|c|c|c}
                           & Appro. ratio & Time & Space & Algorithm \\
\hline
SNN~\cite{Muthukrishnan00} & $O(\lg{N}\lg^*{N})$ & $O(N^{O(1)}+N{\rm polylog}{(N)})$ & $O(N^{O(1)})$ & Offline \\
Shapira and Storer~\cite{Shapira2007} & $O(\lg{N})$ & $O(N^2)$ & $O(N\lg{N})$ & Offline  \\
ESP~\cite{Cormode2007} & $O(\lg{N}\lg^*{N})$ & $O(N\lg^*{N}/\alpha)$ & $N\lg{\sigma}$ & Offline \\
            &                       &                        & $+n(\alpha+3)\lg{(n+\sigma)}$ \\
\hline\hline
OESP & $O(\lg^2{N})$ & $O(\frac{N\lg{N}\lg{n}}{\alpha\lg\lg{n}})$ & $n(\alpha + 1)\lg{(n+\sigma)}$  & Online \\
     &                  &                                         & $+n\lg{(\alpha n)}+5n+o(n)$ 
\end{tabular}
}
\end{center}
\vspace{-1.1cm}
\end{table}

{\em Edit sensitive parsing (ESP)}~\cite{Cormode2007} is an efficient parsing algorithm developed for approximately computing EDM between strings in an offline setting. 
ESP builds from a given string a parse tree that guarantees upper bounds of parsing discrepancies 
between different appearances of the same substring, and then
it represents the parse tree as a vector each dimension of 
which represents the frequency of the corresponding node label 
in a parse tree. 
$L_1$ distance between such characteristic vectors for two strings can approximate the EDM. 
Although ESP has an efficient approximation ratio $O(\lg{N}\lg^*{N})$ and runs fast in $O(N\lg^*{N}/\alpha)$ time for a parameter $\alpha \in (0,1]$ 
for hash tables, its applicability is limited to an offline  case. 
For applications in web mining and Bioinformatics, computing an EDM of massive streaming text data has 
ever been an important task. 
An open challenge, which is receiving increased attention, is to develop a scalable online pattern matching for EDM. 

We present an online pattern matching for EDM. 
Our method is an online version of ESP named {\em online ESP (OESP)}
that (i)~builds a parse tree for a streaming text in an online manner, 
(ii)~computes characteristic vectors for a substring at each position of the streaming text and a query,
and (iii)~computes the $L_1$ distance between each pair of characteristic vectors. 
The working space of our method does not depend on the length of text but the size of a parse tree. 
To make the working space smaller, OESP builds a parse tree from a streaming text and directly encodes it 
into a succinct representation by leveraging the idea behind recent results of an online grammar compression~\cite{Maruyama2013,maruyama2014} 
and a dynamic succinct tree~\cite{NavSad12}. 
Our representation includes a novel succinct representation of a tree named {\em post-order unary degree sequence (POUDS)} that 
is built by the post-order traversal of a tree and a unary degree encoding. 
To guarantee the approximate EDM computed by OESP, 
we also prove an upper bound of the approximation ratio between our approximate EDM and the exact EDM. 

Experiments using standard benchmark texts revealed OESP's efficiencies. 

\section{Preliminaries}
\subsection{Basic notation}
Let $\Sigma$ be a finite alphabet forming texts, and $\sigma=|\Sigma|$. 
$\Sigma^*$ denotes the set of all texts over $\Sigma$, and 
$\Sigma^\ell$ denotes the set of all texts of length $\ell$ over $\Sigma$, i.e. 
$\Sigma^\ell=\{S\in \Sigma^*| |S|=\ell\}$. 
We assume a recursively enumerable set ${\cal X}$ of variables such that $\Sigma \cap {\cal X} = \phi$
and all elements in $\Sigma\cup {\cal X}$ are totally ordered. 
A sequence of symbols from $\Sigma \cup {\cal X}$ is called a string.
The length of string $S$ is denoted by $|S|$, 
and the cardinality of a set $C$ is similarly denoted by $|C|$. 
A pair and triple of symbols from $\Sigma \cup {\cal X}$ are called digram and trigram, respectively. 
Strings $x$ and $z$ are said to be the prefix and suffix of the string $S=xyz$, respectively, 
and $x,y,z$ are called substrings of $S$. 
The $i$-th symbol of $S$ is denoted by $S[i]$ $(1\leq i \leq |S|)$. 
For integers $i$ and $j$ with $1\leq i \leq j \leq |S|$, the substring of $S$ from $S[i]$ to $S[j]$ 
is denoted by $S[i,j]$. 
$N$ denotes the length of a text $S$ and it can be variable in an online setting. 

\subsection{Context-free grammar}
A {\em context-free grammar (CFG)} is a quadruple $G=(\Sigma,V,D,Z_s)$ where $V$ is a finite subset of ${\cal X}$, 
$D$ is a finite subset of $V\times (V\cup \Sigma)^*$ of production rules, 
and $Z_s\in V$ represents the start variable. 
$D$  is also called a {\em phrase dictionary}. 
Variables in $V$ are called nonterminals. 
The set of strings in $\Sigma^*$ derived from $Z_s$ by $G$ is denoted by $L(G)$.
A CFG $G$ is called {\em admissible} if for any $Z\in {\cal X}$ there is exactly one production rule $Z\to \gamma \in D$.
We assume $|\gamma|=2$ or $3$ for any production rule $Z\to \gamma$. 

The parse tree of $G$ is represented as a rooted ordered tree with internal nodes labeled by variables in $V$ and 
leaves labeled by elements in $\Sigma$, and the label sequence of its leaves are equal to an input string. 
Any internal node $Z\in V$ in a parse tree corresponds to a production rule in the form of $Z\to \gamma$ in $D$.
The height of $Z$ is the height of the subtree whose root is $Z$. 

\subsection{Phrase and reverse dictionaries}
For a set $V$ of production rules, 
a {\em phrase dictionary} $D$ is a data structure for directly accessing the 
phrase $S\in (\Sigma \cup V)^*$ for any given $Z\in V$ if $Z\to S\in D$. 
A {\em reverse dictionary} $D^{-1}: (\Sigma \cup V)^*\to V$ 
is a mapping from a given sequence of symbols to a variable.
$D^{-1}$ returns a variable $Z$ associated with a string $S$ if $Z\to S \in D$; 
otherwise, it creates a new variable $Z^\prime \notin V$ and returns $Z^\prime$. 
For example, if $D=\{Z_1\to abc, Z_2\to cd\}$, $D^{-1}(a,b,c)$ returns $Z_1$, 
while $D^{-1}(b,c)$ creates $Z_3$ and returns it.



\subsection{Problem definition}
In order to describe our method we first review the notion of EDM.
The EDM $d(S,Q)$ between two strings $S$ and $Q$ is the minimum number of 
edit operations defined below to transform $S$ into $Q$:
\begin{enumerate}
\item Insertion: A character $a$ at position $i$ in $S$ is inserted, which generates $S[1,i-1]aS[i]S[i+1,N]$,
\item Deletion: A character $a$ at position $i$ in $S$ is deleted, which generates $S[1,i-1]S[i+1,N]$,
\item Replacement: A character at position $i$ is replaced by $a$, which generates $S[1,i-1]aS[i+1,N]$,
\item Substring move: A substring $S[i,j]$ is moved and inserted at the position $k$, which generates $S[1,i-1]S[j+1,k-1]S[i,j]S[k,N]$.
\end{enumerate}

\begin{prob}[Online pattern matching for EDM]\label{prob:1}
For a streaming text $S \in \Sigma^*$, a query $Q \in \Sigma^*$, and a distance threshold $k \geq 0$, 
find all $i \in [1,|S|]$ such that the EDM between a substring $S[i,i+|Q|]$ and $Q$ is at most $k$, i.e. $d(S[i,i+|Q|], Q)\leq k$.
\end{prob}
Cormode and Muthukrishnan~\cite{Cormode2007} presented an offline algorithm for computing EDM. 
In their algorithm, a special type of derivation tree called ESP is constructed for approximately computing EDM.
We present an online variant of ESP.
Our algorithm approximately solves Problem~\ref{prob:1} and is composed of two parts: 
(i) an online construction of a parse tree space-efficiently and 
(ii) an approximate computation of EDM from the parse tree. 
Although our method is an approximation algorithm, it guarantees an upper bound for the exact EDM. 
We now discuss the two parts in the next section. 

\section{Online Algorithm}
OESP builds a special form of CFG and directly encodes it into a succinct representation in an online manner. 
Such a representation can be used as space-efficient phrase/reverse dictionaries, which 
resulted in reducing the working space.
In this section, we first present a simple variant of ESP in order to introduce the notion of {\em alphabet reduction} and {\em landmark}. 
We then detail OESP and approximate computations of the EDM in an online manner. 
In the next section, we present an upper bound of the approximate EDM for the exact EDM.

\subsection{ESP}\label{sec:esp}
Given an input string $S\in\Sigma^*$, we decompose the current $S$ into digrams $WX$ or trigrams $WXY$
associated with variables as production rules, and iterate this process while $|S|>1$ for the resulting $S$.

In each iteration, ESP uniquely partitions $S$ into maximal non-overlapping substrings such that 
$S=S_1S_2\cdots S_\ell$ and each $S_i$ is categorized into one of three types, i.e.,
type1: a repetition of a symbol,
type2: a substring not including a type1 substring and of length at least $\lceil \lg|S|\rceil$,
and type3: a substring being neither type1 nor type2 substrings. 

At one iteration of parsing $S_i$, ESP builds two kinds of subtrees from digram $WX$ and trigram $WXY$, respectively.
The first type is a $2$-tree corresponding to a production rule in the form of $Z\to WX$. 
The second type is a $3$-tree corresponding to $Z\to WXY$.

ESP parses $S_i$ according to its type. 
In case $S_i$ is a type1 or type3 substring, ESP performs the left aligned parsing 
where $2$-trees are built from left to right in $S_i$
and a $3$-tree is built for the last three symbols if $|S_i|$ is odd, as follows:
\begin{itemize} 
  \item If $|S_i|$ is even, ESP builds $Z\to S_i[2j-1,2j]$, $j=1,...,|S_i|/2$,
  \item Otherwise, it builds $Z\to S_i[2j-1,2j]$ for $j=1,...,(\lfloor |S_i|/2 \rfloor-1)$, 
and builds $Z\to  S_i[2j-1,2j+1]$ for $j=\lfloor|S_i|/2\rfloor$.
\end{itemize}
In case $S_i$ is type2, ESP further partitions $S_i=s_1s_2...s_\ell$ ($2\leq |s_j|\leq 3$) 
by the {\em alphabet reduction} described below, and builds $Z\to s_j$ for $j=1,...,\ell$.

After parsing all $S_i$ to $S'_i$, ESP continues this process for the resulted string
by concatenating all $S'_i$ $(i=1,\ldots,\ell)$ at the next level.


{\bf Alphabet reduction:}
Alphabet reduction is a procedure for partitioning a string of type2 into digrams and trigrams.
Given $S$ of type2, consider each $S[i]$ represented as binary integers.
Let $p$ be the position of the least significant bit in which $S[i]$ differs from $S[i-1]$, 
and let ${\it bit}(p,S[i])\in\{0,1\}$ be the value of $S[i]$ at the $p$-th position, where $p$ starts at $0$.
Then, $L[i]=2p+{\it bit}(p,S[i])$ is defined for any $i\geq 2$.
Since $S$ contains no repetition (i.e., $S$ is type2), the string $L$ defined by $L=L[2]L[3]\ldots L[|S|]$ is also type2.
We note that if the number of different symbols in $S$ is $m$, denoted by $[S]=m$, clearly $[L]\leq 2\lg m$.
Then, $S[i]$ is called {\em landmark}
if (i) $L[i]$ is maximal such that $L[i]>\max\{L[i-1],L[i+1]\}$ or 
(ii) $L[i]$ is minimal such that $L[i]<\min\{L[i-1],L[i+1]\}$ and
not adjacent to any other maximal landmark.

Because $L$ is type2 and $[L]\leq \lg |S|$, 
any substring of $S$ longer than $\lg |S|$ must contain at least one landmark.
After deciding all landmarks, if $S[i]$ is a landmark, we replace $S[i-1,i]$ by a variable $X$
and update the current dictionary with $X\to S[i-1,i]$.
After replacing all landmarks, the remaining substrings are replaced by the left aligned parsing.

\subsection{Post-order CFG}
\begin{figure*}[t]
\begin{center}
\includegraphics[width=0.8\textwidth]{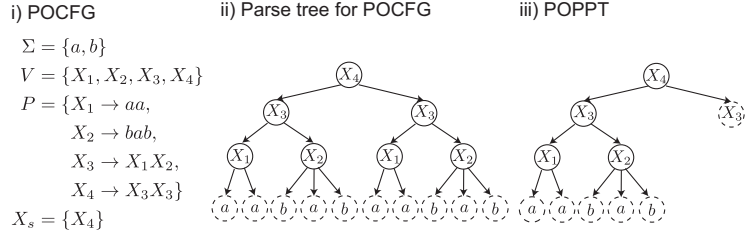}
\end{center}
\vspace{-0.6cm}
\caption{Example of a POCFG, the parse tree of a POCFG, a post-order partial parse tree (POPPT).}
\label{fig:outline}
\vspace{-0.5cm}
\end{figure*}

OESP builds a post-order partial parse tree (POPPT) and directly encodes it into a succinct representation. 
A partial parse tree defined by Rytter~\cite{Rytter03} is the ordered tree formed by 
traversing a parse tree in a depth-first manner and pruning out all descendants under every node 
of nonterminal symbols appearing no less than twice. 
\begin{definition}[POPPT and POCFG~\cite{Maruyama2013}]
A post-order partial parse tree (POPPT) is a partial parse tree 
whose internal nodes have post-order variables. 
A post-order CFG (POCFG) is a CFG whose partial parse tree is a POPPT. 
\end{definition}
Note that the number of nodes in the POPPT is at most $3n$ for a POCFG of $n$ variables, 
because the right-hand sides consist of digrams or trigrams in the production rules 
and the numbers of internal nodes and leaves are $n$ and at most $2n$, respectively. 

Examples of a POCFG and POPPT are shown in Figure~\ref{fig:outline}-i) and iii), respectively. 
The POPPT is built by traversing the parse tree in Figure~\ref{fig:outline}-ii) in depth-first manner and pruning out all the descendants under the node 
having the second $X_3$. The resulted POPPT in Figure~\ref{fig:outline}-iii) consists of internal nodes having post-order variables. 

A major advantage of POPPT is that we can directly encode it into a succinct representation 
which can be used as a phrase dictionary.
Such a representation enables us to reduce the working space of OESP by using it in a combination with a reverse dictionary. 


\subsection{Online construction of a POCFG}
OESP builds from a given input string a POCFG that guarantees upper bounds of parsing discrepancies 
between the same substrings in the string. 
The basic idea of OESP is to (i) start from symbols in an input text, 
(ii) replace as many as possible of the same digrams or trigrams in common substrings 
by the same nonterminal symbols, and 
(iii) iterate this process in a bottom-up manner until it generates a complete POCFG. 
The POCFG is built in an online manner and the POPPT corresponding to it consists of nodes having
two or three children. 

OESP builds two types of subtrees in a POPPT from strings $XY$ and $WXY$. 
The first type is a $2$-tree corresponding to a production rule in the form of $Z\to XY$. 
The second type is a $3$-tree corresponding to a production rule in the form of $Z\to WXY$.

OESP builds a $2$-tree or $3$-tree from a substring of a limited length.
Let $u$ be a string of length $m$. 
A function ${\mathcal L}: (\Sigma \cup V)^m \times [m]\to \{0,1\}$ classifies whether or not 
the $i$-th position of $u$ has a landmark, i.e., 
the $i$-th position of $u$ has a landmark if ${\mathcal L}(u,i)=1$. 
${\mathcal L}(u,i)$ is computed from a substring $u[i-1,i+2]$ of length four. 
OESP builds a $3$-tree from a substring $u[i+1,i+3]$ of length three 
if the $i$-th position of u does not have a landmark; 
otherwise, it builds a $2$-tree from a substring $u[i+2,i+3]$ of length two. 
The landmarks on a string are decided such that 
they are synchronized in long common subsequences to make the parsing discrepancies as small as possible. 

\begin{algorithm}[t]
{\scriptsize
\caption{Online construction of ESP.
$D$ is phrase dictionary, $D^{-1}$ is reverse dictionary, and $q_k$ is queue at level $k$.} 
\label{algo:OESP}
}
{\scriptsize
\begin{algorithmic}[1]
\Function{\sc OESP}{}
\State {$D:=\emptyset$; initialize queues $q_k$}
\While{reading a new character $c$ from an input text}
\State{{\sc ProcessSymbol}($q_1$, $c$)}
\EndWhile
\EndFunction
\Function{{\sc ProcessSymbol}}{$q_k$, $X$}
\State $q_k.enqueue(X)$
\If {$q_k.size() = 4$}
\If {${\mathcal L}(q_k,2) = 0$} \Comment{Build a $2$-tree}
\State $Z:=D^{-1}(q_k[3],q_k[4]);$ $D:=D\cup \{Z\to q_k[3]q_k[4]\}$
\State {\sc ProcessSymbol}$(q_{k+1}, Z)$
\State $q_k.dequeue()$; $q_k.dequeue()$
\EndIf
\ElsIf {$q_k.size() = 5$} \Comment{Build a $3$-tree}
\State $Z := D^{-1}(q_k[3],q_k[4],q_k[5])$; $D := D \cup \{Z \rightarrow q_k[3]q_k[4]q_k[5]\}$
\State {\sc ProcessSymbol}$(q_{k+1},Z)$
\State $q_k.dequeue()$; $q_k.dequeue()$; $q_k.dequeue()$
\EndIf
\EndFunction
\end{algorithmic}
}
\end{algorithm}

The algorithm uses a set of queues, $q_k, k=1,...,m$, where $q_k$ processes
the string at $k$-th level of a parse tree of a POCFG and builds $2$-trees and $3$-trees at each $k$.
Since OESP builds a balanced parse tree, the number $m$ of these queues is bounded by $\lg{N}$. 
In addition, landmarks are decided on strings of length at most four, and the length of each queue is also fixed to five.
Algorithm~\ref{algo:OESP} consists of the functions {\sc OESP} and {\sc ProcessSymbol}. 

The main function is {\sc OESP} which reads new characters from an input text and 
gives them to the function {\sc ProcessSymbol} one by one. 
The function {\sc ProcessSymbol} builds a POCFG in a bottom-up manner. 
There are two cases according to whether or not a queue $q_k$ has a landmark. 
For the first case of ${\mathcal L}(q_k,2)=0$, i.e. $q_k$ does not have a landmark, 
the $2$-tree corresponding to a production rule $Z\to q_k[3]q_k[4]$ in a POCFG 
is built for the third and fourth elements $q_k[3]$ and $q_k[4]$ of the $k$-th queue $q_k$.
For the other case, the $3$-tree corresponding to a production rule $Z\to q_k[3]q_k[4]q_k[5]$ is built for the third, fourth 
and fifth elements $q_k[3]$, $q_k[4]$ and $q_k[5]$ of the $k$-th queue $q_k$. 
In both cases, the reverse dictionary $D^{-1}$ returns a nonterminal symbol replacing a sequence of symbols. 
The generated symbol $Z$ is given to the higher $q_{k+1}$,
which enables the bottom-up construction of a POCFG in an online manner. 

The computation time and working space depend on implementations of phrase and reverse dictionaries. 
The phrase dictionary for a POCFG of $n$ variables can be implemented using a standard array of at most $3n\lg{(n+\sigma)}$ bits of space and $O(1)$ access time. 
In addition, the reverse dictionary can be implemented using a chaining hash table and a phrase dictionary implemented as an array. 
Thus, the working space of OESP using these data structures is at most $n(4+\alpha)\lg{(n+\sigma)}$ bits. 
In the following subsections, we present space-efficient representations of phrase/reverse dictionaries.

\subsection{Compressed phrase dictionary}
\begin{figure*}[t]
\begin{center}
\includegraphics[width=0.6\textwidth]{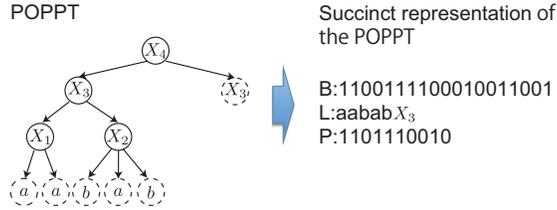}
\end{center}
\vspace{-0.6cm}
\caption{Succinct representation of a POCFG for a phrase dictionary.}
\label{fig:pouds}
\vspace{-0.5cm}
\end{figure*}

OESP directly encodes a POCFG into a succinct representation
that consists of bit strings $B$, $P$ and a label sequence $L$. 
A bit string $B$ is built by traversing a POPPT and putting $c$ $0$s and $1$ for a node having $c$ children 
in the post-order. The final $0$ in $B$ represents the super node. 
We shall call the bit string representation of a POPPT {\em posterior order unary degree sequence~(POUDS)}. 
To dynamically build a tree and access any node in the POPPT, we index $B$ by using the {\em dynamic range min/max tree}~\cite{NavSad12}. 
Our POUDS supports two tree operations: $child(B,i,j)$ returns the $j$-th child of a node $i$; 
$num\_child(B,i)$ returns the number of children for a node $i$. 
They are computed in $O(\lg{m}/\lg\lg{m})$ time while using $2m+o(1)$ bits of space 
for a tree having $m$ nodes.

A bit string $P$ is built by traversing a POPPT and putting $1$ for a leaf and $0$ for an internal node 
in the post-order.  
$P$ is indexed by the rank/select dictionary~\cite{Jacobson89,Navarro12}.
The label sequence $L$ stores symbols of leaves in a POPPT. 

We can access any element in $L$ as a child of a node $i$ in the following. 
First, we compute $c=num\_child(B,i)$ and children nodes $p=child(B,i,j)$  for $j\in [1,c]$. 
Then, we can compute the positions in $L$ corresponding to the positions of these children as $q=rank_1(P,p)$ that returns the number of occurrences of $1$ in $P[0,p]$ in $O(1)$ time.
We obtain leaf labels as $L[q]$.
For a POCFG of $n$ nonterminal symbols, 
we can access the right-hand side of symbols from the left-hand side of a symbol of a production rule 
in $O(\lg{n}/\lg\lg{n})$ time while using at most $n\lg{(n+\sigma)}+5n+o(n)$ bits of space.

\subsection{Compressed reverse dictionary}
We implement a reverse dictionary using a chaining hash table that has a load factor $\alpha \in (0,1]$ in a combination with a phrase dictionary.
The hash table has $\alpha n$ entries and each entry stores a list of integers $i$ 
representing the left-hand side $X_i$ of a rule. 
For the rule $X_i\to S$, the hash value is computed from the right-hand side $S$.
Then, the list corresponding to the hash value is scanned to search for $X_i$ while checking elements referred to as $S$ in a phrase dictionary.
Thus, the expected access time is $O(1/\alpha)$. 
The space for a POCFG with $n$ nonterminal symbols is $\alpha n\lg(n+\sigma)$ bits for the hash table and $n\lg(n+\sigma)$ bits for the lists, 
which resulted in $n(\alpha + 1)\lg(n+\sigma)$ bits in total. 

A crucial observation in OESP is that indexes $i$ for nonterminal symbols $X_i$ are created in a strictly increasing order. 
Thus, we can organize each list in a hash table as a strictly increasing sequence of the indexes of nonterminal symbols. 
We insert a new index $i$ into a list in the hash table, and we append it at the end of the list. 
Each list in the hash table consists of a strictly increasing sequence of indexes. 
To make each index smaller, we compute the difference between an index $i$ and 
the previous one $j$, and we encode it by the delta code, 
which resulted in the difference $i-j$ being encoded in 
$1+\lfloor \lg{(i-j)} \rfloor + 2\lfloor \lg \lfloor 1+\lg{(i-j)\rfloor}\rfloor$ bits.
For all $n$ nonterminal symbols, the space for the lists is upper bounded 
by $n(1+\lg{(\alpha n)} + 2\lg\lg{(\alpha n)})$. bits 
The space for the hash table is $\alpha n \lg{(n+\sigma} + n(1+\lg{(\alpha n)} + 2\lg\lg{(\alpha n)})$ bits in total, 
resulting in $\alpha n\lg(n+\sigma)+n(1+\lg(\alpha n))$ bits by multiplying the original $\alpha$ by a constant. 

Since the reverse dictionary is implemented using the chaining hash and the phrase dictionary, 
its total space is at most $n(\alpha + 1)\lg{(n+\sigma)}+n(5+\lg{(\alpha n)})+o(n)$ bits.
We can obtain the following result.
\begin{lemma}~\label{lem:1}
For a string length $N$, OESP builds a POCFG of n nonterminal symbols and its phrase/reverse dictionaries 
in $O(\frac{N\lg{n}}{\alpha\lg\lg{n}})$ expected time using at most $n(\alpha + 1)\lg{(n+\sigma)}+n\lg{(\alpha n)}+5n+o(n)$ bits of space. 
\end{lemma}

\subsection{Online pattern matching with EDM}
We approximately solve problem~\ref{prob:1} by using OESP. 
First, the parse tree is computed from a query $Q$ by OESP. 
Let $T(Q)$ be a set of node labels in the parse tree for $Q$. 
We then compute a vector $V(Q)$ each dimension $V(Q)(e)$ of which 
represents the frequency of the corresponding node label $e$ in $T(Q)$.

OESP builds another parse tree for a streaming text $S$ in an online manner. 
$T(S)[i,i+|Q|]$ is a set of node labels included in the subtree 
corresponding to a substring $S[i,i+|Q|]$ from $i$ to $i+|Q|$ in $T(S)$. 
$V(S)[i,i+|Q|]$ can be constructed for each $i \in [1,|S|-|Q|]$ by 
adding the node labels corresponding to $S[i,i+|Q|]$ and
subtracting the node labels not included in $T(S)[i,i+|Q|]$ from $V(S)[i,i+|Q]]$, 
which can be performed in $\lg{|S|}$ time.

$L_1$-distance approximates the EDM between $V(S)[i,i+|Q|]$ and $V(Q)$, and 
it is computed as $||V(S)[i,i+|Q|]-V(Q)||=\sum_{e \in (T(S)[i,i+|Q|] \cup T(Q))}|V(S)[i,i+|Q|](e) - V(Q)(e)|$. 
We obtain the results with respect to computational time and space for computing the $L_1$ distance 
from lemma~\ref{lem:1} as follows.
\begin{theorem}~\label{thr:1}
For a streaming text $S$ of length $N$, 
OESP approximately solves the problem~\ref{prob:1} in $O(\frac{N\lg{N}\lg{n}}{\alpha\lg\lg{n}})$ expected time using at most $n(\alpha + 1)\lg{(n+\sigma)}+n\lg{(\alpha n)}+5n+o(n)$ bits of space. 
\end{theorem}

\begin{figure}[t]
\begin{center}
\begin{tabular}{cc}
dna.200MB & english.200MB \\
\includegraphics[width=0.3\textwidth]{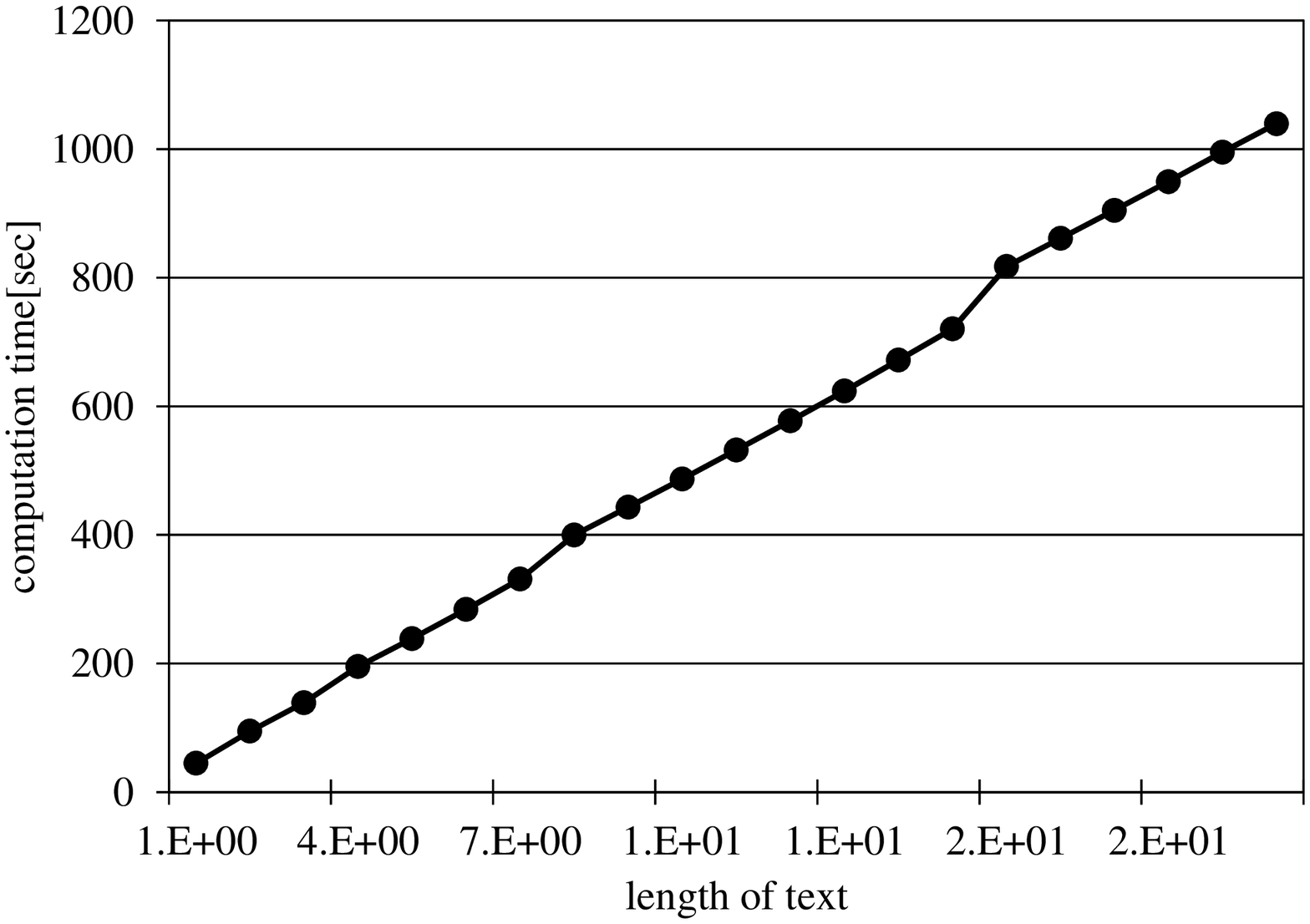} &
\includegraphics[width=0.3\textwidth]{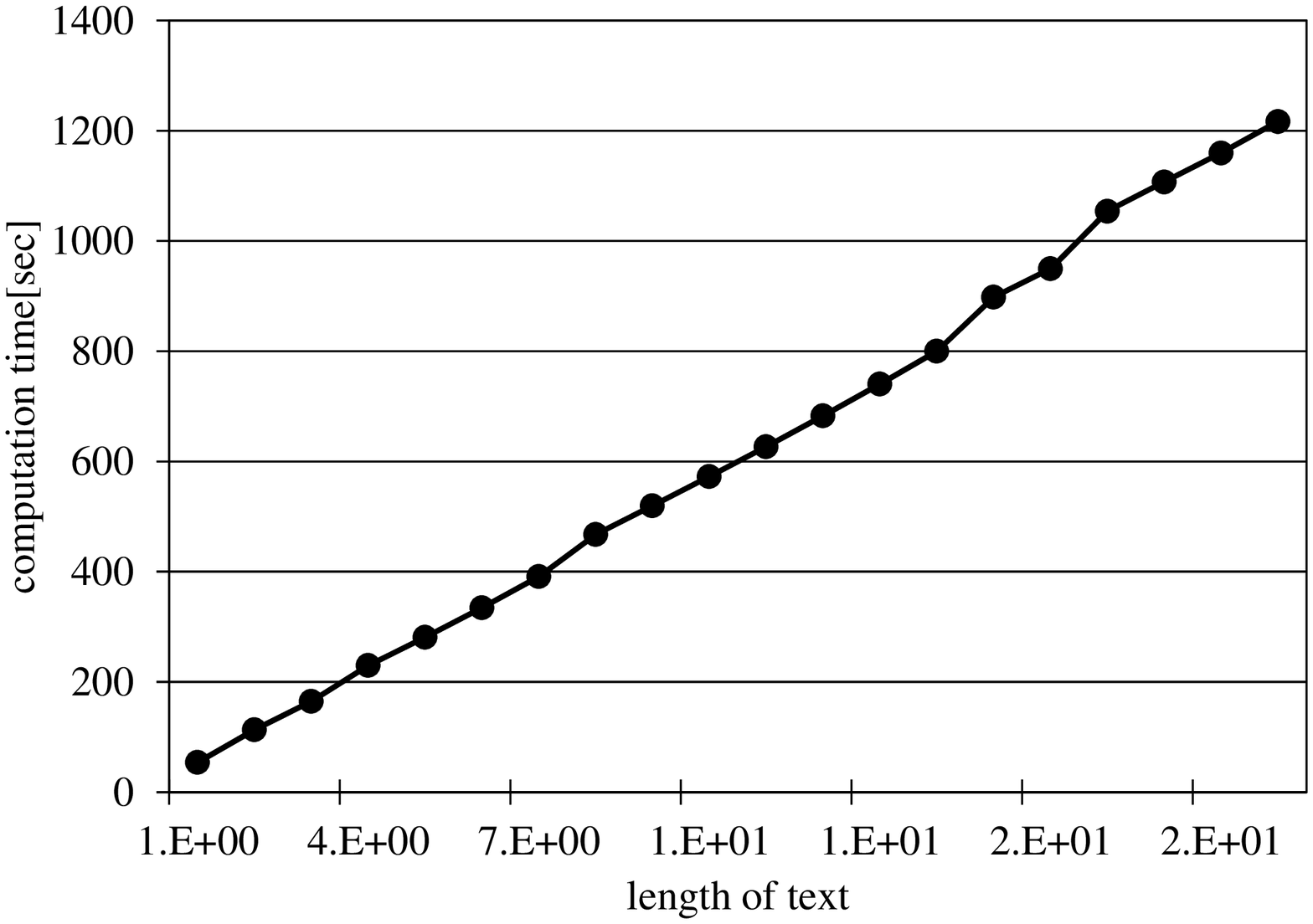} 
\end{tabular}
\end{center}
\vspace{-0.8cm}
\caption{Computation time in seconds for the length of text.}
\label{fig:time}
\vspace{-0.5cm}
\end{figure}

\section{Upper Bound of Approximation}
We present an upper bound of the approximate EDM in this section.

\begin{theorem}\label{upper-bound}
$||V(S)-V(Q)|| = O(\lg^2{m})d(S,Q)$ 
for any $S,Q\in \Sigma^*$ and $m=\max\{|S|,|Q|\}$.
\end{theorem}

\begin{proof}
Let $e_1,e_2,\ldots,e_d$ be a shortest series of editing operations such that 
$S_{k+1} = S_k(e_k)$ where $S_1 = S$, $S_d(e_d) = Q$, and $d=d(S,Q)$.
It is sufficient to prove the assumption: there exists a constant $c$ such that 
$||V(S)-V(Q)|| \leq c\lg^2m$ for $R(e)=S$.
$S(i)$ denotes the string resulted by the $i$-th iteration of ESP where $S(0)=S$.
Let $p_i,q_i$ be the smallest integers satisfying $S(i)[p_i]\neq Q(i)[p_i]$ and
$S(i)[|S|(i)-q_i]\neq Q(i)[|Q(i)|-q_i]$, respectively.
We show that $q_i-p_i \leq \lg m+1$ for each height $i$.
This derives $||V(S)-V(Q)|| \leq 2\lg m(\lg m +1)$ because $i\leq \lg m$.

We begin with the case that $e$ is an insertion of a symbol.
Clearly, it is true for $i=0$ since $q_0-p_0\leq 1$.
We assume the hypothesis on some height $i$.
Let $S(i)[p']$ be the closest landmark from $S(i)[p_i]$ with $p'<p_i$ and
$S(i)[q']$ be the closest landmark from $S(i)[q_i]$ with $q_i<q'$.
For the next height, let $S(i+1)=S_1S_2S_3$ such that the tail of $S_1$ derives $S(i)[p_i]$
and the tail of $S_2$ derives $S(i)[q_i]$, and let $Q(i+1)=Q_1Q_2Q_3$ such that $|Q_1|=|S_1|$ and $|Q_3|=|S_3|$.
On any iteration of ESP, the left aligned parsing is performed from a landmark to its closest landmark.
It follows that, for $S_1$, $S_1[j]=Q_1[j]$ except their tails,
for $S_2$, $|S_2| \leq \lfloor\frac{1}{2}(q_i-p_i) \rfloor \leq \lfloor\frac{1}{2}(\lg m+1)\rfloor$, 
and for $S_3$, we can estimate $S_3[j]=Q_3[j]$ for any $j>\lfloor\frac{1}{2}\lg m\rfloor$.
Thus, $q_{i+1}-p_{i+1}\leq 1 + \lfloor\frac{1}{2}(\lg m+1)\rfloor +
\lfloor\frac{1}{2}\lg m\rfloor \leq \lg m+1$. 
Since $d(S,Q)=d(Q,S)$, this bound is true for the deletion of any symbol.
The case that $e$ is a replacement is similar.

Moreover, the bound holds for the case of insertion or deletion of any string of length at most $\lg m$.
Using this, we can reduce the case of move operation of a substring $u$ as follows.
Without loss of generality, we assume $u$ is a type2 substring and let 
$u=xyz$ such that $x/z$ are the shortest prefix/suffix of $u$ that contain a landmark, respectively.
Then, we note that the $y$ inside of $u$ is transformed to a same string
for any occurrence of $u$.
Therefore, the case of moving $u$ from $S$ to obtain $Q$ is reduced to 
the case of deleting $x,z$ at some positions and inserting them into other positions.
Since $|x|,|z|\leq \lg m$, the case of moving $u$ is identical to the case of
inserting two symbols and deleting two symbols, i.e., $||V(S)-V(Q)|| \leq 8\lg m(\lg m +1)$.
\end{proof}

From theorem~\ref{thr:1} and~\ref{upper-bound}, 
we obtain the following main theorem.

\begin{theorem}~\label{main-theorem}
EDM is $O(\lg^2 N)$-approximable by the proposed online algorithm 
with $O(\frac{N\lg{N}\lg{n}}{\alpha\lg\lg{n}})$ expected time
and $n(\alpha + 1)\lg{(n+\sigma)}+n(5+\lg{(\alpha n)})+o(n)$ bits of space.
\end{theorem}
\begin{proof}
By the theorem~\ref{upper-bound},
we obtain the bound $||V(S[i,i+|Q|]-V(Q)||=O(\lg^2|Q|)d(S[i,i+|Q|],Q)$ for any $i \in [1,|S|-|Q|]$.
The time complexity is proved by the theorem~\ref{thr:1}.
Thus, for the strings $S$ and $Q$ with $N=|S|\geq |Q|$, the result is concluded.
\end{proof}
 

\begin{figure}[t]
\begin{center}
\begin{tabular}{cc}
dna.200MB & english.200MB \\
\includegraphics[width=0.3\textwidth]{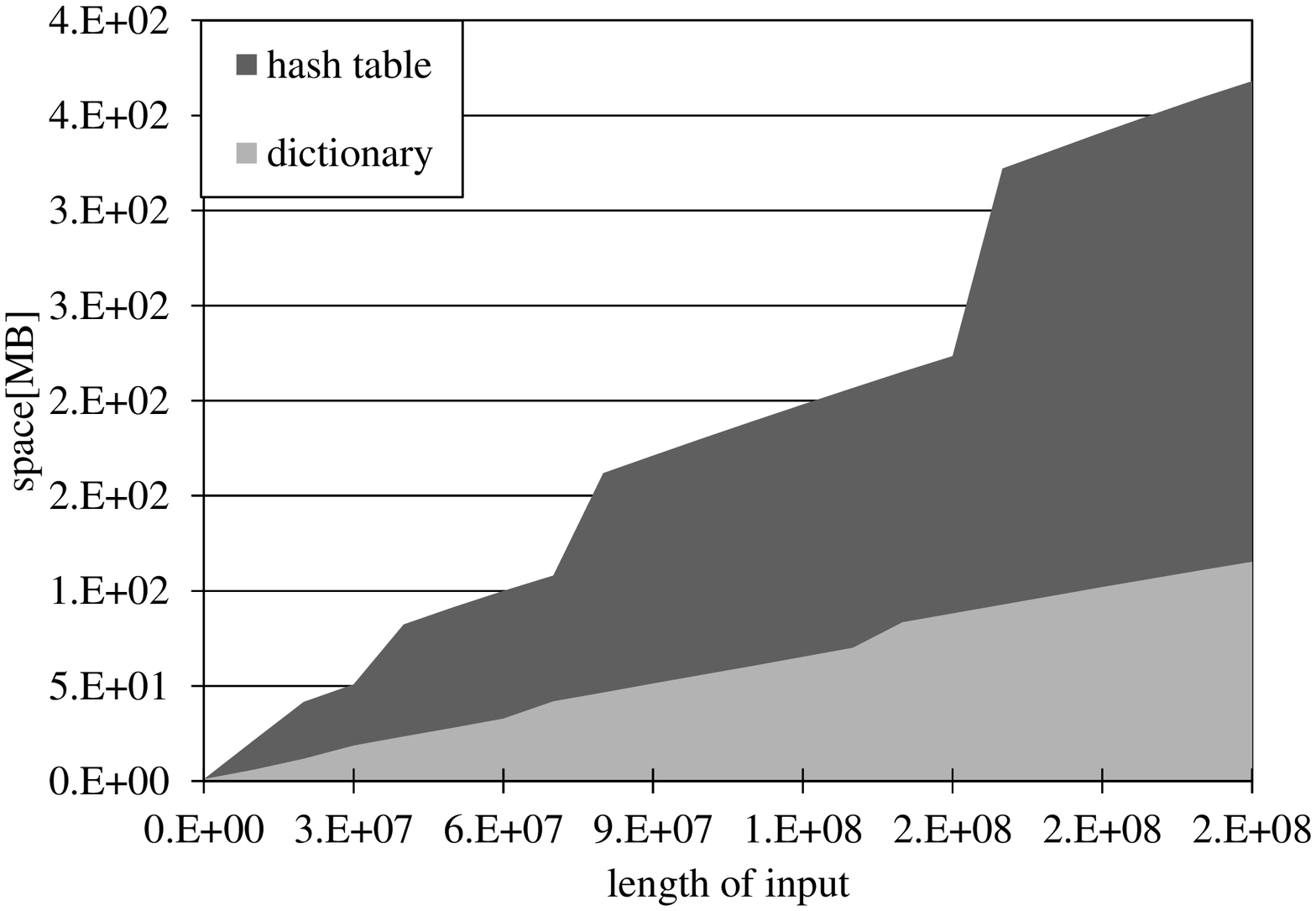} &
\includegraphics[width=0.3\textwidth]{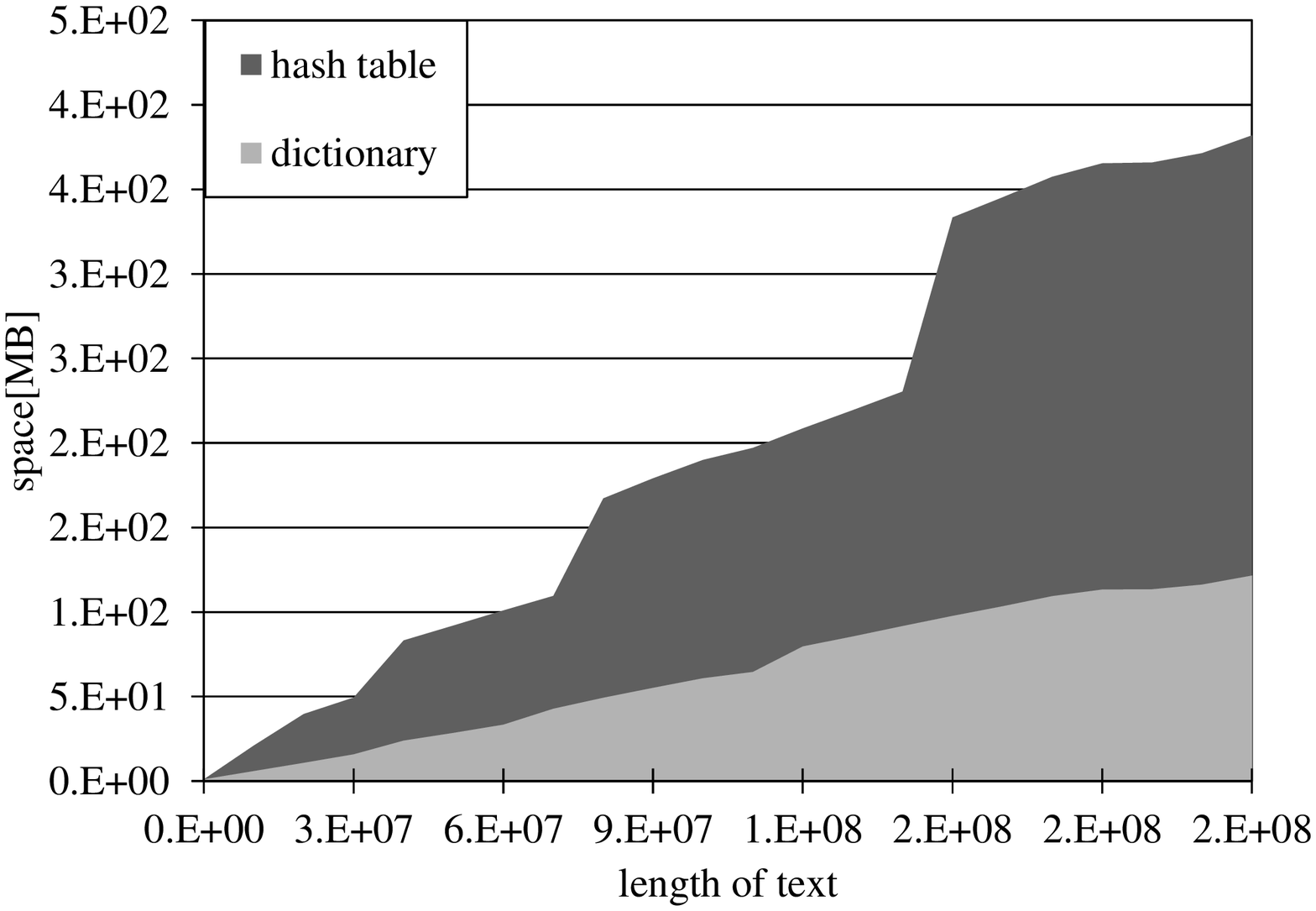} 
\end{tabular}
\end{center}
\vspace{-0.8cm}
\caption{Working space of dictionary and hash table for the length of text.}
\label{fig:compow}
\vspace{-0.2cm}
\begin{center}
\begin{tabular}{cc}
dna.200MB & english.200MB \\ 
\includegraphics[width=0.3\textwidth]{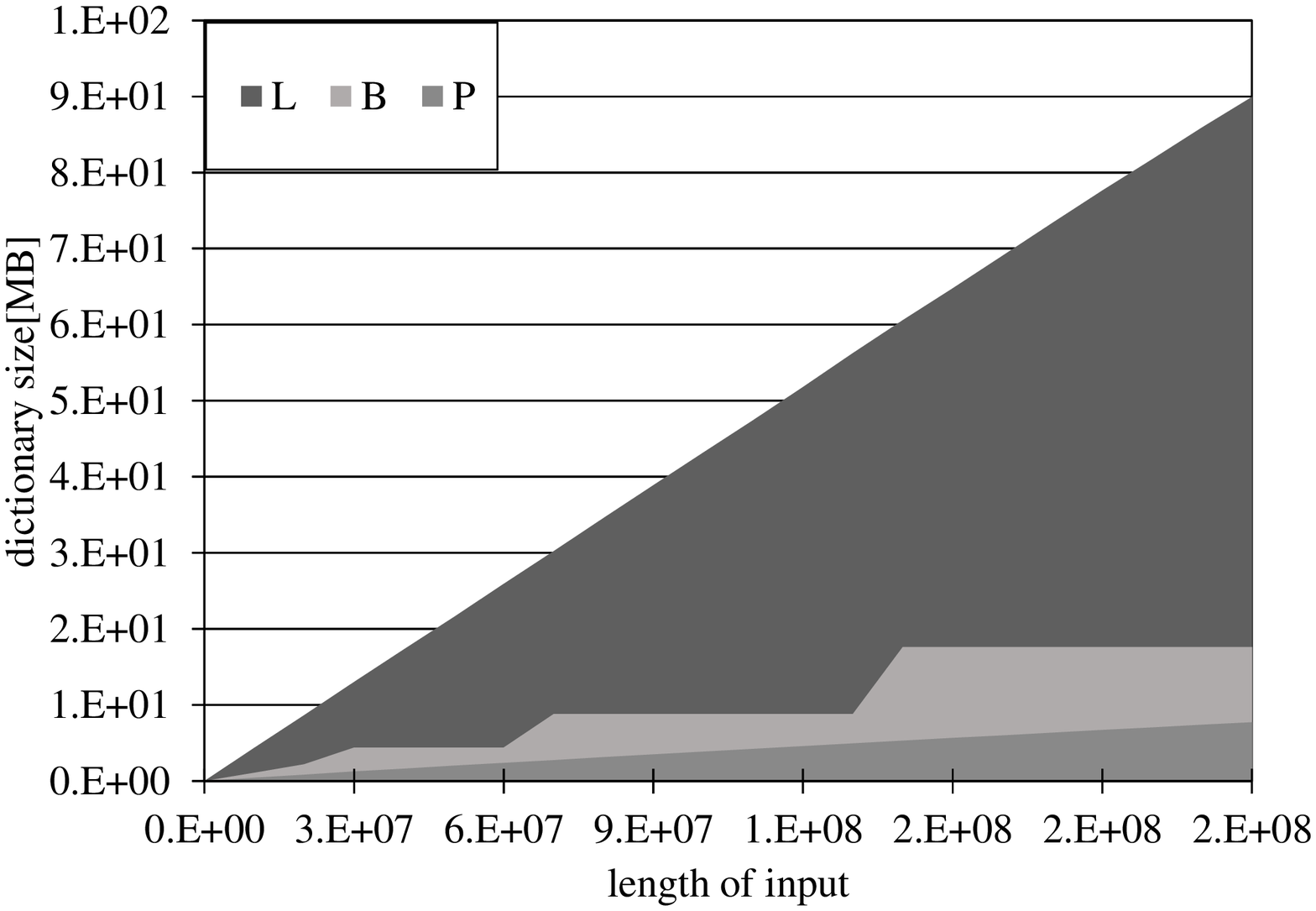} &
\includegraphics[width=0.3\textwidth]{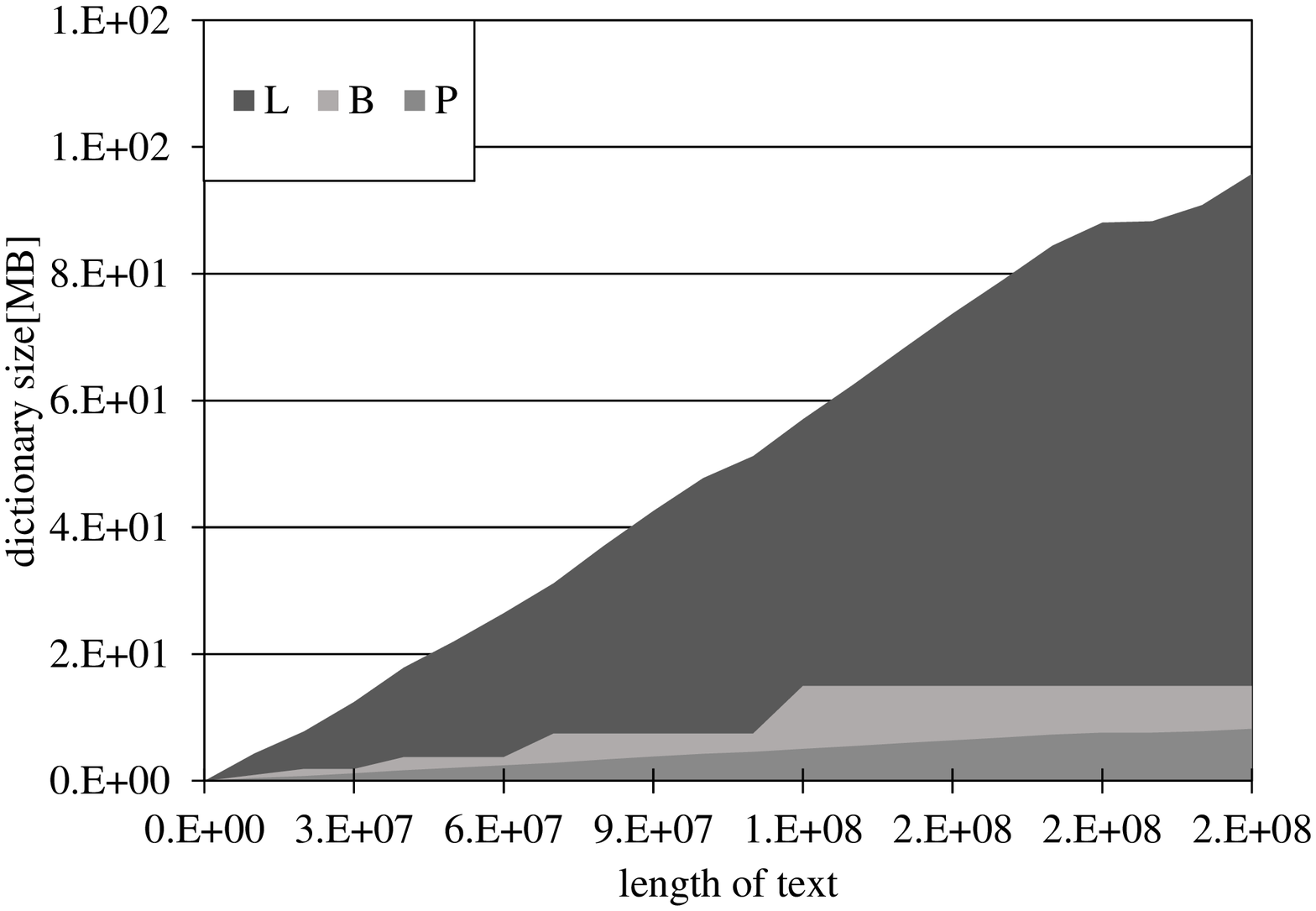} 
\end{tabular}
\end{center}
\vspace{-0.8cm}
\caption{Working space of a POUDS~(B), a label sequence~(L) and a bit string~(P) which organizes a dictionary.}
\label{fig:compod}
\end{figure}
\begin{figure}[t]
\begin{center}
\begin{tabular}{cc}
dna.200MB & english.200MB \\
\includegraphics[width=0.3\textwidth]{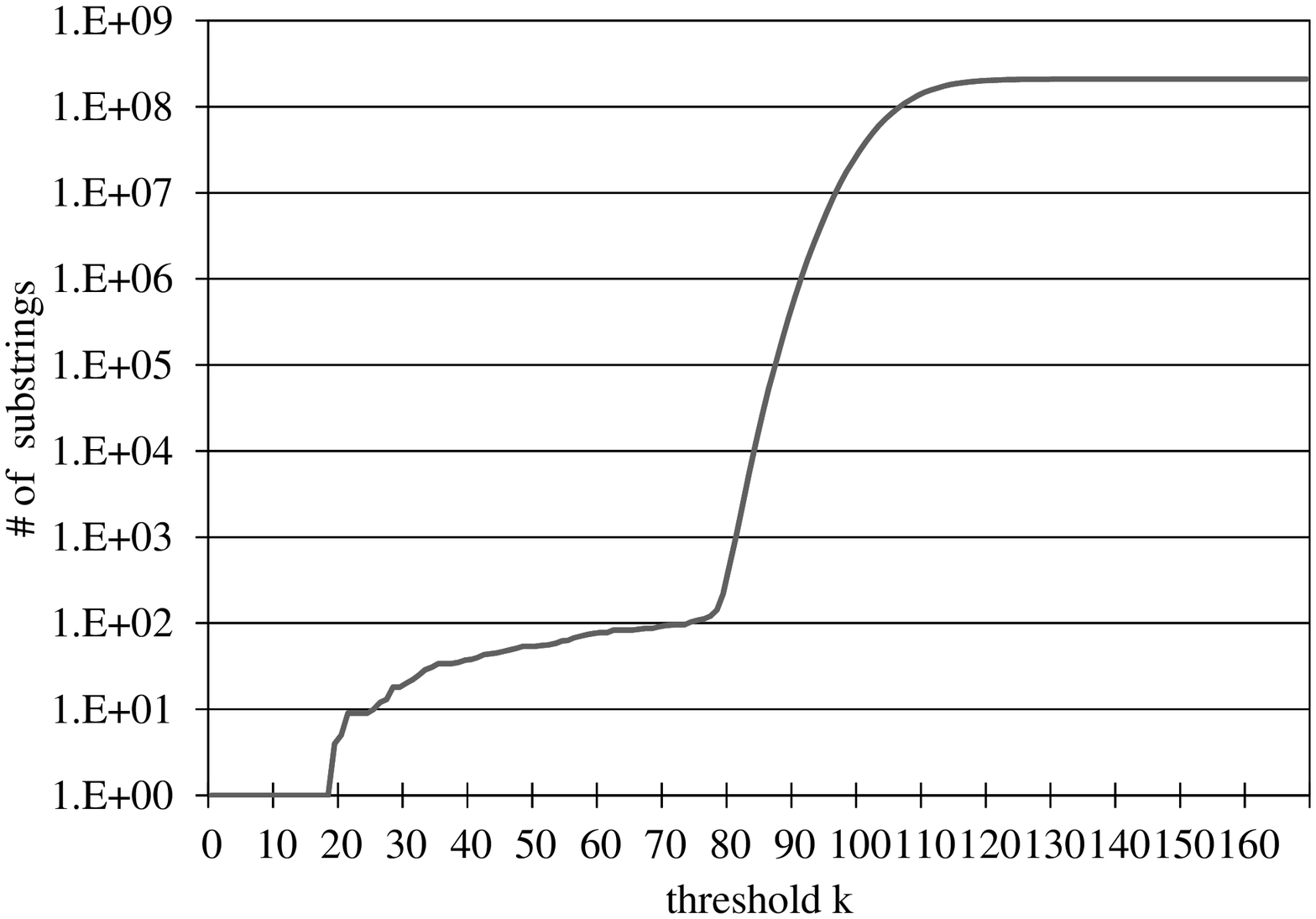} &
\includegraphics[width=0.3\textwidth]{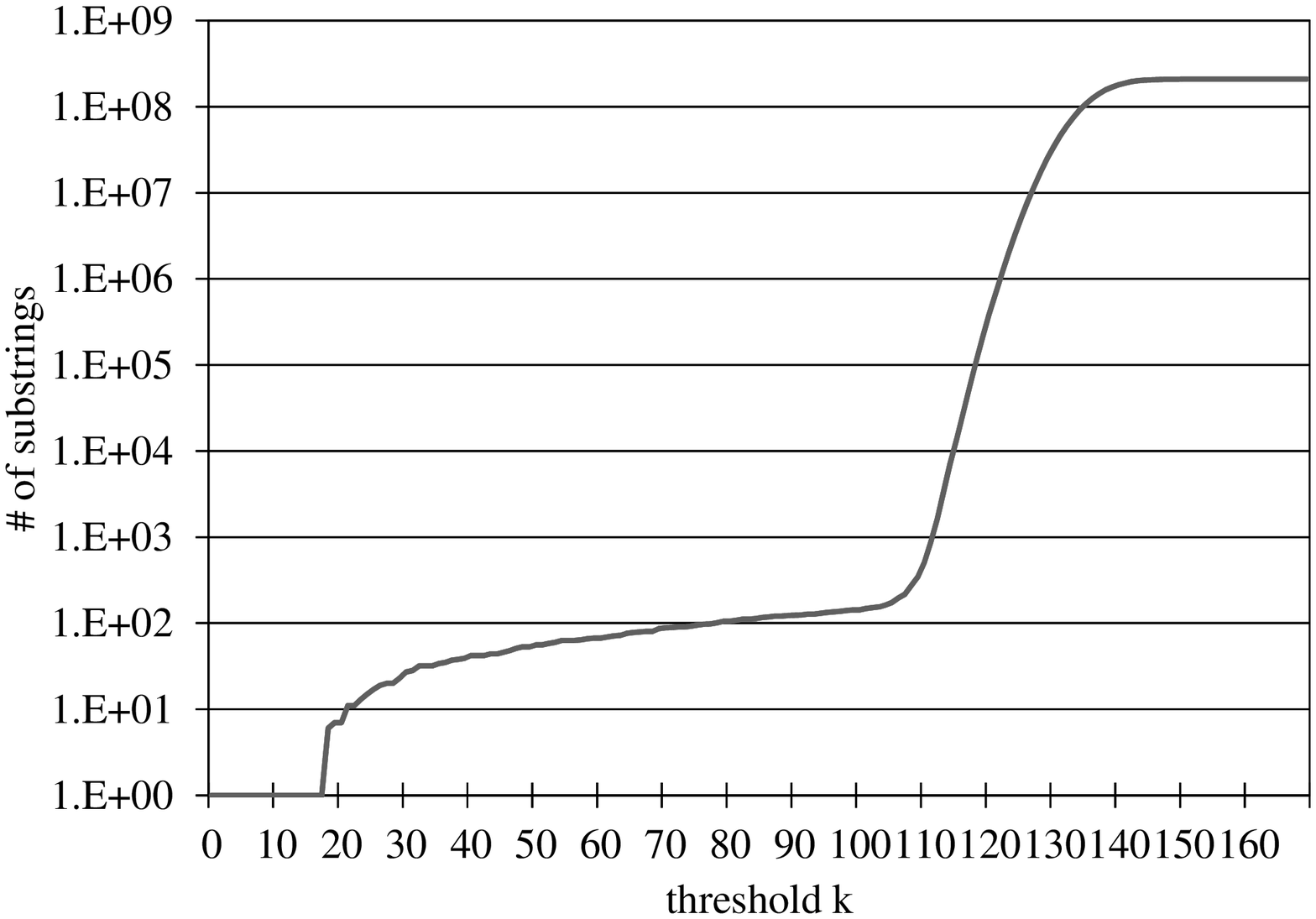} 
\end{tabular}
\vspace{-0.2cm}
\end{center}
\vspace{-0.6cm}
\caption{The number of substrings whose EDM to a query is no more than each threshold.}
\label{fig:thre}
\vspace{-0.6cm}
\end{figure}

\begin{table}[t]
\begin{center}
{\scriptsize
  \caption{Space for POUDS $B$, label sequence $P$ and bit string $P$ organizing a dictionary on dna.200MB and english.200MB.} \label{tab:dic}
}
\vspace{-0.3cm}
{\scriptsize
  \begin{tabular}{r|c|c|c}
                      & L[MB] & B[MB] & P[MB]  \\ \hline
dna.200MB     &$ 89.95$  & $17.62$ & $7.73$ \\
english.200MB &$95.72$  & $14.99$ &  $8.22$ \\
  \end{tabular}
}
\end{center}
\vspace{-1.0cm}
\end{table}

\section{Experiments}
We evaluated OESP on one core of an eight-core Intel Xeon CPU E7-8837 (2.67GHz) machine with 1024GB memory. 
We used two standard benchmark texts dna.200MB and english.200MB downloadable from \url{http://pizzachili.dcc.uchile.cl/texts.html}.
We sampled texts of length $100$ from these texts as queries. 
We also used computation time and working space as evaluation measures. 

Figure~\ref{fig:time} shows computation time for increasing the length of text. 
The computation time increased linearly for the length of text. 

Figure~\ref{fig:compow} shows working space for increasing the length of text. 
The space of dictionary was much smaller than that of hash table. 
The dicionary used $115$MB for dna.200MB and $121$MB for english.200MB, 
while the hash table used $368$MB for dna.200MB and $382$MB for english.200MB. 

Figure~\ref{fig:compod} shows space of a POUDS, a label sequence and a bit string organizing a dictionary for increasing the length of text. 
The space of dictionary and bit string was much smaller than that of the label sequence for dna.200MB and english.200MB.
Table~\ref{tab:dic} details those space. 

Figure~\ref{fig:thre} shows the number of substring whose EDM to a query is at most a threshold. 
There were thresholds where the number of substrings dramatically increases. 
The results showed the applicability of OESP to streaming texts.

\section{Conclusion}
We have presented an online pattern maching for EDM. 
Our method named OESP is an online version of ESP.
A future work is to apply OESP to real world streaming texts. 

\nocite{JalseniusPS13}
\nocite{CliffordS11}
\nocite{Clifford09}
\nocite{Clifford08}
\nocite{Porat09}
\nocite{Clifford13}
\nocite{Clifford12}

{\small
\bibliographystyle{plain}
\bibliography{compress,succinct,cpm,book,other}
}

\end{document}